\documentclass{IEEEtran}
 
\usepackage{cite}

\usepackage{amsmath,amssymb,amsfonts}
\usepackage{amsthm}
\usepackage{algorithm}
\usepackage{algpseudocode}
\usepackage{graphicx}
\usepackage{textcomp}
\usepackage{empheq}
\usepackage{mathrsfs} 
\usepackage{physics} 
\usepackage{derivative} 
\usepackage{cuted} 
\usepackage{calc}
\usepackage{xparse}
\usepackage{multicol}
\usepackage{subcaption}
\usepackage{booktabs} 

\usepackage{lipsum}
\usepackage{microtype}
\usepackage{optidef}
\usepackage[per-mode=symbol]{siunitx}
\usepackage[hidelinks]{hyperref}
\usepackage[nameinlink,capitalise]{cleveref}

\AfterEndEnvironment{strip}{\leavevmode}

\newtheorem{lem}{Lemma}

\newcommand{\bigO}{\mathcal{O}}

\newcommand{\tveeG}{^{\scriptscriptstyle\vee}} 
\newcommand{\twedG}{^{\scriptscriptstyle\wedge}}
\newcommand{\lieg}{\mathfrak{g}}
\newcommand{\LieG}{\mathrm{G}}
\newcommand{\AdG}{\mathrm{Ad}}
\newcommand{\adG}{\mathrm{ad}}
\newcommand{\expG}{\, \mathrm{exp}}
\newcommand{\logG}{\, \mathrm{log}}
\newcommand{\LambdaG}{\mathrm{J}}
\newcommand{\LambdaGbar}{ \bar{\mathrm{J}}}

\newcommand{\inv}{^{-1}}
\newcommand{\tp}{^{ \sf T}}
\newcommand{\invtp}{^{-\sf T}}

\DeclareMathOperator{\diag}{diag}
\DeclareMathOperator{\Ad}{Ad}
\DeclareMathOperator{\ad}{ad}

\newtheorem*{ass*}{Assumptions}
\newtheorem*{problemstatement*}{Problem Statement}

\newtheorem{thm}{Theorem}
\newtheorem{prop}{Proposition}
\newcommand{\Log}{\mathrm{Log}}
\newcommand{\Exp}{\mathrm{Exp}}

\newcommand{\Rf}{\mathbb{R}}

\newcommand{\GL}{\mathrm{GL}}

\newcommand{\SO}{\mathrm{SO}}
\newcommand{\so}{\mathfrak{so}}

\newcommand{\SE}{\mathrm{SE}}

\newcommand{\nhphantom}[1]{\sbox0{#1}\hspace{-\the\wd0}}

\DeclareFontFamily{U}{mathx}{\hyphenchar\font45}
\DeclareFontShape{U}{mathx}{m}{n}{<-> mathx10}{}
\DeclareSymbolFont{mathx}{U}{mathx}{m}{n}
\DeclareMathAccent{\widecheck}{0}{mathx}{"71}
\DeclareFontFamily{U}{mathx}{\hyphenchar\font45}
\DeclareFontShape{U}{mathx}{m}{n}{<-> mathx10}{}
\DeclareSymbolFont{mathx}{U}{mathx}{m}{n}
\DeclareMathAccent{\widehat}{0}{mathx}{"70}
\makeatletter
\newcommand\makebig[2]{%
  \@xp\newcommand\@xp*\csname#1\endcsname{\bBigg@{#2}}%
  \@xp\newcommand\@xp*\csname#1l\endcsname{\@xp\mathopen\csname#1\endcsname}%
  \@xp\newcommand\@xp*\csname#1r\endcsname{\@xp\mathclose\csname#1\endcsname}%
}
\makeatother
\makebig{biggg} {3.0}
\makebig{Biggg} {3.5}
\makebig{bigggg}{4.0}
\makebig{Bigggg}{5.5}

\newlength\mytemplena
\newlength\mytemplenb
\DeclareDocumentCommand\myalignalign{sm}
{
  \settowidth{\mytemplena}{$\displaystyle #2$}%
  \setlength\mytemplenb{\widthof{$\displaystyle=$}/2}%
  \hskip-\mytemplena%
  \hskip\IfBooleanTF#1{-\mytemplenb}{+\mytemplenb}%
}

\Crefname{equation}{}{}
\Crefname{figure}{Fig.}{Figs.}
\Crefname{tabular}{Tab.}{Tabs.}
\crefname{ass}{assumption}{assumptions}

\usepackage{colortbl}
\definecolor{lightgray}{gray}{0.9}
\newcommand{\greycell}{\cellcolor{lightgray}}

\allowdisplaybreaks

\makeatletter
\renewcommand{\maketag@@@}[1]{\hbox{\m@th\normalsize\normalfont#1}}%
\makeatother

\newcommand{\MYfooter}{\smash{\scriptsize
\hfil\parbox[t][\height][t]{\textwidth}{\centering
This work has been submitted to the IEEE for possible publication. Copyright may be transferred without notice, after which this version may no longer be accessible.}\hfil\hbox{}}}

\makeatletter

\def\ps@headings{%
\def\@oddhead{\MYfooter}
\def\@evenhead{\MYfooter}
\def\@oddfoot{}%
\def\@evenfoot{}}

\def\ps@IEEEtitlepagestyle{%
\def\@oddhead{\MYfooter}%
\def\@evenhead{\MYfooter}%
\def\@oddfoot{}%
\def\@evenfoot{}}

\makeatother

\begin{document}


\title{Equivalence of Left- and Right-Invariant  Extended Kalman Filters on Matrix Lie Groups}

\author{Finn G. Maurer, Erlend A. Basso, Henrik M. Schmidt-Didlaukies and Torleiv H. Bryne. 
\thanks{This work was supported by the Research Council of Norway through the Resilience and Reliability Improvement of CORS Network Based Services for Automated and Autonomous Transportation Operations (REACTOR) project (No. 344275) and the European Research Council (ERC) under the European Union's Horizon 2020 research and innovation programme, through the ERC Advanced Grant 101017697-CRÈME.}
\thanks{The authors are with the Department of Engineering Cybernetics, Norwegian University of Science and Technology, NO-7491 Trondheim, Norway {\tt \scriptsize finn.g.maurer@ntnu.no}}}

\maketitle

\begin{abstract}
This paper derives the extended Kalman filter (EKF) for continuous-time systems on matrix Lie groups observed through discrete-time measurements. By modeling the system noise on the Lie algebra and adopting a Stratonovich interpretation for the stochastic differential equation (SDE), we ensure that solutions remain on the manifold.  
The derivation of the filter follows classical EKF principles, naturally integrating a necessary full-order covariance reset post-measurement update. A key contribution is proving that this full-order covariance reset guarantees that the Lie-group-valued state estimate is invariant to whether a left- or right-invariant error definition is used in the EKF. Monte Carlo simulations of the aided inertial navigation problem validate the invariance property and confirm its absence when employing reduced-order covariance resets.

\end{abstract}
\section{Introduction}

The extended Kalman filter (EKF) has been the industry standard for state estimation problems with nonlinear dynamics and/or measurement equations. The recursive algorithm is essential in several applications, including navigation systems, signal processing, robotics, and finance.  At its core, the EKF is designed to estimate the state of a dynamic system by recursively incorporating measurements and predictions while minimizing the impact of sensor noise and dynamic uncertainty.

Matrix Lie groups are differentiable manifolds that are also groups. For a multitude of engineering systems, we can identify their configuration space with a matrix Lie group, such as the orientation of a rigid body. Matrix Lie groups provide a natural framework for incorporating group-specific constraints and symmetries into the filtering process. 
Matrix Lie groups are also desirable due to their computational advantages, since they allow for efficient implementation of group operations, such as exponentials, logarithms and group multiplication.

While the classical formulation of an EKF is done in Euclidean space \cite{mendel_lessons_1995,maybeck_stochastic_1982}, using a Lie-group-based EKF has yielded superior state estimation performance for engineering systems with a configuration space that can be identified with a matrix Lie group. 
An early example of such a filter is \cite{lefferts_kalman_1982}, where an EKF is developed to estimate spacecraft orientation using quaternions. 

A discrete EKF for matrix Lie groups (D-LG-EKF) was introduced in \cite{bourmaud_discrete_2013} and the continuous-discrete equivalent (CD-LG-EKF) in \cite{bourmaud_continuous-discrete_2015}. The multiplicative EKF (MEKF) is derived in \cite{markley_attitude_2003} and more thoroughly developed and analyzed in \cite{sola_quaternion_2017}, where it is presented under the name error state KF (ESKF).  It uses a unit quaternion for global orientation presentation and a three-component orientation error representation. The invariant EKF (IEKF) was developed in \cite{barrauInvariantExtendedKalman2017, barrau_invariant_2018, bonnabel_left-invariant_2008}, designed for systems with specific invariance properties and where the dynamics are so-called group-affine. They show that such systems have state-independent linear error propagation and consistency properties for systems that are not fully observable. 

The filters exhibit distinct approaches in their interpretation of stochastic differential equations (SDEs). Specifically, the CD-LG-EKF employs the Itô interpretation, whereas the MEKF and IEKF utilize the Stratonovich interpretation. As suggested by \cite{van1981ito,kloeden_numerical_1992,Moon_2014,oksendal_stochastic_2003}, the Stratonovich interpretation is preferred when Gaussian white noise, as represented in the system's differential equation, is considered an idealization of an actual physical (colored) noise process. The Stratonovich interpretation is also advantageous because it maintains the chain rule of ordinary calculus and results in a less complex covariance propagation. Moreover, when the state space is a matrix Lie group embedded in Euclidean space and the noise is modeled on the Lie algebra, the Stratonovich interpretation is the only viable approach, as the Itô interpretation would cause the solution to the SDE to deviate from the matrix Lie group submanifold \cite{barrauIntrinsicFilteringLie2015}.

The filters also exhibit varying degrees of generalizability. The CD-LG-EKF is formulated to accommodate any state and measurement on matrix Lie groups. In contrast, the IEKF is specifically derived for group affine systems and measurements that possess invariance properties. The MEKF, on the other hand, is explicitly tailored for estimation problems pertaining to orientation.


Another significant distinction lies in the application (or omission) of a \emph{covariance reset} during the correction step. Since the correction step pertains to the error state, reparameterizing the posterior distribution is essential to transfer information from the error state back to the nominal state, resulting in a covariance reset. The CD-LG-EKF employs a full-order reset, while the MEKF, as presented in \cite{markley_attitude_2003}, omits the reset and the MEKF presented in \cite{sola_quaternion_2017} uses a first-order approximation.  In contrast, the IEKF omits the reset step entirely, which can be interpreted as a zero-order approximation, thereby simplifying the proof of the filter's convergence properties. The study by \cite{gill_full-order_2020} assesses various reset strategies for the orientation estimation problem and concludes that a full-order reset yields superior estimation performance.

Furthermore, works introducing the IEKF and the MEKF also discuss the use of left- or right-invariant estimation errors as a key design consideration. This raises the question of which formulation is more appropriate for different estimation problems. According to \cite{barrau_three_2017}, a left-invariant error should be employed when the measurements are left-invariant, and analogously, a right-invariant error is preferable for right-invariant measurements. Solà \cite{sola_quaternion_2017} presents the MEKF using both left- and right-invariant errors (referred to as local and global errors), and suggests that the MEKF based on right-invariant errors exhibits superior performance. Numerous experimental studies \cite{de_araujo_robust_2023, hartley_contact-aided_2019, potokar_invariant_2021, cohen_navigation_2020, tang_invariant_2023, xu_invariant_2024} suggest that the selection of error parameterization plays a significant role when a reduced-order covariance reset is used. 

The main contributions of this paper are twofold. Firstly, we derive the EKF for continuous-time systems on matrix Lie groups with discrete-time Euclidean measurements. In contrast to \cite{bourmaud_continuous-discrete_2015}, we adopt the Stratonovich interpretation of an SDE and we denote the resulting filter as the Stratonovich continuous-discrete Lie group EKF, or simply the LEKF. Consequently, the continuous-time covariance propagation equation is greatly simplified compared to \cite{bourmaud_continuous-discrete_2015}. Our derivation is based on the derivation of the classical EKF and results in filter equations that naturally include a covariance reset step after each measurement update.
Notably, if the covariance reset step was omitted and if the system was group affine, the filter equations would become identical to an IEKF.

Secondly, we prove that the Lie-group-valued state estimate obtained from the LEKF is invariant with respect to the choice of estimation error. Specifically, LEKFs that employ either left-invariant or right-invariant estimation errors yield identical state estimates. Consequently, there is no inherent advantage in selecting one type of estimation error over the other, except for considerations related to implementation. For this result to hold, it is essential to utilize the full-order covariance reset, as opposed to a finite number of terms in the series representation, which is a common practice in the well-known MEKF and IEKF methodologies.
A Monte Carlo simulation demonstrates the equivalence of left- and right-invariant filters when full-order resets are employed and reveals that this equivalence is lost with reduced-order resets. Furthermore, the Monte Carlo simulation indicates that filters employing the full-order reset exhibit the smallest mean average error.

This paper is organized as follows. \Cref{sec:preliminiaries} introduces the notation, terminology and mathematical preliminaries that will be used throughout this paper. Subsequently, \Cref{sec:ekf-derivations} presents the derivation of the EKF for continuous-time systems on matrix Lie groups with discrete-time measurements. \Cref{sec:ekf-equivalence} details the principal theoretical contribution of this work, demonstrating that the left- and right-invariant LEKFs are equivalent in the sense that they both yield the same state estimate and that the covariance matrix is related through a change of variables.
In \Cref{sec:case-study}, we validate this equivalence through Monte Carlo simulations, additionally highlighting that this equivalence does not hold when a reduced-order covariance reset is applied. Finally, \Cref{sec:conclusions-and-future-work} provides the concluding remarks and outlines potential directions for future research.

\section{Preliminaries}\label{sec:preliminiaries}
In this section, we introduce the mathematical notation and concepts that will be utilized throughout the paper.

\subsection{Notation}
Throughout this paper, the topology considered is the standard Euclidean topology on $\mathbb{R}^n$ or $\mathbb{R}^{m \times n}$. The general linear group is denoted by $\GL(n) \coloneqq \{ g \in \mathbb{R}^{n \times n} : \det(g) \ne 0 \}$. The matrix exponential, $\Exp(\cdot)$, is defined via its power series expansion, and $\Log(\cdot)$ denotes the matrix logarithm, taken as the inverse of the exponential. The identity matrix of size $n$ is denoted by $I_n$. The notation $\bigO(\cdot)$ refers to Big-O, describing standard asymptotic behavior. To keep the notation concise, functions that depend on an exogenous variable $u$ are written as $a_u(\cdot)$ instead of $a(u, \cdot)$ in algorithms and proofs.

    
\subsection{Matrix Lie Groups}
A matrix Lie group $\LieG$ is defined by two properties \cite{hall_lie_2015}:
\begin{enumerate}
    \item $\LieG$ is a relatively closed subset of $\GL(n)$.
    \item $\LieG$ is a matrix subgroup of $\GL(n)$.
\end{enumerate}
According to this definition, every matrix Lie group is a properly embedded submanifold of $\GL(n)$. 
By a relatively closed subset of $\GL(n)$, we mean a set that can be written as $S \cap \GL(n)$ for some closed set $S \subset \Rf^{n \times n}$. If $\LieG$ is a closed subset of $\GL(n)$, then it is also a relatively closed subset of $\GL(n)$. A matrix subgroup $\LieG$ of $\GL(n)$ is a set of invertible $n \times n$ matrices such that $g \in \LieG$ implies $g\inv \in \LieG$ and $g_{1}, g_{2} \in \LieG$ implies $g_{1} g_{2} \in \LieG$. 

The Lie algebra of a matrix Lie group $\LieG$ is denoted $\lieg$, and defined as the real vector space 
\begin{align}
    \lieg \coloneqq \{ X \in \Rf^{n \times n} : t \in \Rf \implies \exp(X t) \in \LieG  \}
\end{align}
equipped with the matrix commutator $[X,Y] \coloneqq XY - Y \mskip-2mu X $. The dimension $k$ of a matrix Lie group $\LieG$ as a manifold is equal to the dimension of $\lieg$ as a vector space. Since $\lieg$ is a real vector space of dimension $k$, there exists an isomorphism $(\, \cdot \,)\twedG  : \Rf^{k} \to \lieg$ with inverse $(\, \cdot \,)\tveeG : \lieg \to \Rf^{k}$. For $g \in \LieG$ and $\xi, \zeta \in \Rf^{k}$, we define the adjoint operators $\AdG : \LieG \times \Rf^{k} \to \Rf^{k}$, $\AdG_{g} \xi \coloneqq (g  \xi  \twedG  g\inv )\tveeG $, and $\adG : \Rf^{k} \times \Rf^{k} \to \Rf^{k}$, $\adG_{\zeta}\xi \coloneqq [ \zeta \twedG,  \xi \twedG ]\tveeG$. We also define $\exp(\mu) \coloneqq \Exp( \mu \twedG)$ and $\log(g) \coloneqq   \Log(g)  \tveeG$. 


\subsection{Derivatives and Group Jacobians}

For matrix Lie groups $\LieG$ of dimension $k$, a mapping $f : \LieG \to \Rf^m$ is said to be smooth if
\begin{equation}
    \xi \mapsto  f(g \expG(  \xi))
\end{equation}
is smooth on a neighborhood of $\xi = 0$ for every $g \in \LieG$. The right derivative $\odif{f} : \LieG \to \mathbb{R}^{m \times k}$ is then defined as
\begin{equation}
\odif{f}(g)  = \frac{\partial}{\partial \xi} \bigg(   f(g \expG(  \xi)) \bigg) \bigg|_{\xi = 0} .
\end{equation}

If $f$ is smooth, then the derivative $\odif{f}$ serves as a local approximation to $f$ in the sense that for fixed $g \in \LieG$, we have $f(g \expG(  \xi  ) ) = f(g)+ \odif{f} (g) \xi  + \bigO( \lvert \xi \rvert^{2} )$. A left derivative can be analogously defined by
\begin{align}\label{eq:right_derivative}
    \bar{\odif*{}} f(g)   = \frac{\partial}{\partial \xi} \bigg(   f(\expG(  \xi)  g  ) \bigg) \bigg|_{\xi = 0} , 
\end{align} 
such that \cite{sola_micro_2021}
\begin{equation}
\label{eq:relation-right-left-diff}
    \bar{\odif*{}}f(g) = \odif f(g) \AdG_{g\inv}. 
\end{equation}
For a smooth function $f : \LieG \times \Rf^{l} \to \Rf^{m}$, we denote by ${\odif*{}}_{1} f$ and $\bar{\odif*{}}_{1} f$ the right and left derivatives of $f$ with respect to its first argument.

The right group Jacobian can be equivalently stated as \cite{bullo_proportional_1995}
\begin{align}
    \begin{aligned}
        \LambdaG_{\zeta} \coloneqq&\,\, \frac{\partial}{\partial \xi} \big( \logG( \expG ( - \zeta ) \expG ( \zeta + \xi ) ) \big)\big\rvert_{\xi = 0}\\ 
        =& \int_{0}^{1} \AdG_{\expG(- s\zeta)} \, \mathrm{d}s 
        = \sum_{k = 0}^{\infty} \frac{(-1)^{k}}{(k+1)!} (\adG_{\zeta})^{k}. \label{eq:right-jacobian}
    \end{aligned}
\end{align}
The left group Jacobian can be stated as
\begin{align}
        \begin{aligned}
            \LambdaGbar_{\zeta} \coloneqq&\,\, \frac{\partial}{\partial \xi} \big( \logG( \expG (  \zeta + \xi ) \expG ( -\zeta  ) ) \big)\big\rvert_{\xi = 0} \\
        =& \int_{0}^{1} \AdG_{\expG(  s\zeta)} \, \mathrm{d}s 
        =  \sum_{k = 0}^{\infty} \frac{1}{(k+1)!} (\adG_{\zeta})^{k}.  \label{eq:left-jacobian}
        \end{aligned}
\end{align}
Clearly, $\LambdaGbar_{\zeta} = \LambdaG_{-\zeta}$. It is also possible to give representations for the inverse of $\LambdaG_{\zeta}$, which is the derivative of the logarithm map. We have
\begin{align}
\label{eq:log-likelihood}
    \begin{aligned}
        \LambdaG_{\zeta}\inv \coloneqq&\,\, \frac{\partial}{\partial \xi} \big( \logG( \expG (  \zeta   ) \expG ( \xi ) ) \big)\big\rvert_{\xi = 0}\\
        =&   \sum_{k = 0}^{\infty} \frac{(-1)^{k} B_{k} }{k!} (\adG_{\zeta})^{k},
        \end{aligned}
\end{align}
and similarly
\begin{align}
    \begin{aligned}
        \LambdaGbar_{\zeta} \inv \coloneqq&\,\,  \frac{\partial}{\partial \xi} \big( \logG( \expG ( \xi  ) \expG (  \zeta  )  ) \big)\big\rvert_{\xi = 0} \\
        =&  \sum_{k = 0}^{\infty} \frac{  B_{k} }{k!} (\adG_{\zeta})^{k} ,
        \end{aligned}
\end{align}
where $B_{k}$ denotes the Bernoulli numbers with $B_{0} = 1$, $B_{1} = -\tfrac{1}{2}$, $B_{3} = \tfrac{1}{6}$, $B_{4} = - \tfrac{1}{30}$, and so on. The following proposition establishes two key properties of the group Jacobians. 
\begin{prop}\label{prop:jac}
    It holds that 
    \begin{align}
    \AdG_{\expG(\zeta)} \LambdaG_{\zeta} &= \LambdaG_{ -\zeta} , \label{eq:Adexpxi-Lambdaxi} \\ 
     \AdG_{g}  \LambdaG_{ \zeta } &= \LambdaG_{ \AdG_{g} \zeta }  \AdG_{g} . \label{eq:Adg-Lambda}
\end{align}
\end{prop}
\begin{proof}
    We have
    \begin{align*}
        \AdG_{\expG (\zeta)} \LambdaG_{\zeta} &= \AdG_{\expG (\zeta)} \int_{0}^{1} \AdG_{\expG(- s\zeta)} \, \mathrm{d}s \\
        & \hspace{-0.5cm}= \int_{0}^{1} \AdG_{\expG( (1 - s) \zeta)} \, \mathrm{d}s 
        = \int_{0}^{1} \AdG_{\expG( \bar{s} \zeta)} \, \mathrm{d} \bar{s} = \LambdaG_{-\zeta},
    \end{align*}
    where we performed the substitution $\bar{s} = 1 - s$. To show the second claim, note that
    \begin{align*}
        \AdG_{g} \LambdaG_{\zeta} &=   \AdG_{g} \int_{0}^{1} \AdG_{\expG(- s\zeta)} \, \mathrm{d}s \\
        &= \int_{0}^{1} \AdG_{g} \AdG_{\expG(  - s  \zeta)} \AdG_{g}\inv  \, \mathrm{d}s \, \AdG_{g} \\
        &= \int_{0}^{1} \AdG_{g \expG( - s \zeta) g\inv} \, \mathrm{d} s \, \AdG_{g} \\
        &= \int_{0}^{1} \AdG_{ \expG( - s \AdG_{g} \zeta)  } \, \mathrm{d} s \, \AdG_{g} 
        = \LambdaG_{\AdG_{g} \zeta} \AdG_{g}.  \tag*{\qedhere} 
    \end{align*}
\end{proof}

\subsection{Extended Concentrated Gaussians} 
Following \cite{ge_equivariant_2022} and \cite{ge_geometric_2024}, we define a random variable $g$ on a Lie group $\LieG$ using an extended concentrated Gaussian distribution defined by
\begin{align}
    g &\coloneqq g_o \expG (\xi) \label{eq:body-cdg},
\end{align}
where $g_o\in \LieG$ is a noise-free variable referred to as the reference point, and 
$\xi \sim \mathcal{N}(\mu, P)$ is a $k$-dimensional Gaussian where $P$ is $k\times k$ positive definite covariance matrix, expressed in the body frame. Note that the probability density function (PDF) of $g$, which we denote $p(g) = \mathrm{N}(g;g_o, \mu, P)$, is induced from $\mathcal{N}(\mu,P)$ which is defined on the Lie algebra $\lieg$. In the case where $\mu=0$, the distribution of $g$ becomes just the concentrated Gaussian. In that case, the (right) mean and covariance of $g$ correspond to $g_o$ and $P$, respectively, as derived in \cite{barfoot_associating_2014}. 

The extended concentrated Gaussian can also be defined by expressing the local error distribution in the spatial frame:
\begin{equation}
    g = \expG(\bar \xi)g_o, \label{eq:spatial-cdg}
\end{equation}
where $\bar\xi \sim \mathcal{N}(\bar \mu, \bar P) $. We denote the spatial PDF by $p(g)=\bar{\mathrm{N}}(g;g_o, \bar\mu, \bar P)$. The spatial representation is related to the body representation by $\bar\mu = \AdG_{g_o}\mu$ and $\bar P = \AdG_{g_o}P\AdG_{g_o}\tp$. 
The following Lemma is a spatial counterpart of \cite[Lemma 1]{ge_geometric_2024}.
\begin{lem}\label{lem:reset-spatial}
    Given an extended concentrated Gaussian distribution $p(g) = \bar{\mathrm{N}}(g; x_1, \bar \mu_1, \bar \Sigma_1)$ on $\LieG$ and a point $x_2 \in \LieG$, then the concentrated Gaussian $q(g) = \bar{\mathrm{N}}(g;  x_2, \bar \mu_2, \bar \Sigma_2)$ with parameters
    \begin{align}
        \bar \mu_{2} &= \log ( \exp(\bar \mu_1)x_1 x_2\inv),\\
        \bar \Sigma_2 &= \bar{\mathrm{J}}_{\bar \mu_2}\inv \bar{\mathrm{J}}_{\bar \mu_1} \bar \Sigma_1 \bar{\mathrm{J}}_{\bar \mu_1}\tp \bar{\mathrm{J}}_{\bar \mu_2}\invtp,
    \end{align}
    minimizes the Kullback-Leibler divergence between $p(g)$ and $q(g)$ up second-order linearization error.  
\end{lem} 
The proof follows the exact same steps as \cite[Lemma 1]{ge_geometric_2024}, using \eqref{eq:spatial-cdg} instead of \eqref{eq:body-cdg} and using the left group Jacobian.

\section{Extended Kalman Filters on Matrix Lie Groups}\label{sec:ekf-derivations}
In this section, we derive the EKF for continuous-time systems defined on matrix Lie groups.
We adopt the Stratonovich interpretation of an SDE as it allows us to define the system noise on the Lie algebra, while ensuring that the solution of the SDE does not leave the matrix Lie group. Our derivation of the filter is grounded in foundational works on extended Kalman filtering, with a particular emphasis on the necessity of the covariance reset step. 

\subsection{System Dynamics and Measurements}
Let the system state $g$ be an element of the matrix Lie group $G$ of dimension $k$. The stochastic system dynamics and measurements are given by
\begin{subequations}
\label{eq:stochastic-system}
\begin{align}
        \dot{g} &= g\underbrace{[a(g,u) + B(g,u) \mu]\twedG}_{V^b}
     \!=\! \underbrace{[\bar{a}(g,u) + \bar{B}(g,u) \mu]\twedG}_{V^s} g \label{eq:sde-g-dot}\\
    y &= c(g, u) + D(g,u)\eta
\end{align}
\end{subequations}
where $u \in \mathbb{R}^l$ is the input, $\mu\in\mathbb{R}^s$ is a Gaussian white process with a covariance matrix $Q$, $\eta \in \Rf^{r}$ is Gaussian noise with a covariance matrix $N$ and $y \in \Rf^m$ is a discrete-time measurement. We assume that $a : \LieG \times \Rf^{l} \to \Rf^{k}$, $B : \LieG \times \Rf^{l} \to \Rf^{k \times s}$, $c : \LieG  \times \Rf^l \to \Rf^m$  and $D: \LieG \times \Rf^l \to \Rf^{m\times r}$ are smooth mappings. The stochastic system can be expressed using the body velocity $V^b\in\lieg$ or the spatial velocity $V^s\in \lieg$ \cite{bulloProportionalDerivativePD1995}. 
It follows immediately from the above that
\begin{align}
    \begin{aligned}
     \bar{a}(g,u) &= \AdG_{g} a(g,u), \quad 
     \bar{B}(g,u) = \AdG_{g} B(g,u) .
    \end{aligned}
\end{align}
In this article, we interpret the SDE \eqref{eq:sde-g-dot} in the Stratonovich sense. 
Note that the SDE \eqref{eq:sde-g-dot} is interpreted according to Itô calculus in \cite{bourmaud_continuous-discrete_2015}, where they also assume that  $B(g,u) = I$. However, the Itô interpretation may cause the solution of the SDE to depart from the submanifold $\LieG$, as elaborated in \cite{barrauIntrinsicFilteringLie2015}. 

\subsection{Deriving EKFs on Matrix Lie Groups}
The left LEKF will be derived using the left-invariant error and the body velocity. The derivation of the filter is similar to the classical EKF, as presented in\cite[Lesson 24]{mendel_lessons_1995} and \cite[Chapter 9.5]{maybeck_stochastic_1982}. We define the nominal system model as
\begin{subequations}\label{eq:nominal-state-sys}
\begin{align}
    \dot g_n &\coloneqq g_n[a(g_n, u)]\twedG, \label{eq:nominal-state-propagation}\\
    y_{n} &\coloneqq c(g_{n}, u), 
\end{align}
\end{subequations}
where the input $u$ is assumed to be known. The left-invariant state error is defined as  $g_e\coloneqq g_n\inv g$ and the measurement error as $y_e \coloneqq y- y_n$. 
The dynamics of the state error follow directly from \eqref{eq:sde-g-dot} and \eqref{eq:nominal-state-propagation}, 
\begin{align}\label{eq:error-dynamics}
    \begin{aligned}
        \dot{g}_{e} &= g_{e}[ a(g_{n} g_{e}, u)  -  \AdG_{g_{e}} \inv a(g_{n}, u)]\twedG \\
        &\quad + g_{e}[B(g_{n}g_{e}, u) \mu ]\twedG.
    \end{aligned}
\end{align}
Assuming that $g_{e}$ is sufficiently close to the identity, a local error state can be defined by $\xi \coloneqq \log_\LieG (g_e)$ such that $g = g_n \expG(\xi)$. Utilizing the chain rule and \eqref{eq:log-likelihood}, we find that it evolves according to
\begin{align}\label{eq:local-error-dynamics-devl}
    \begin{aligned}
        \dot{\xi} 
        &= \LambdaG_{\xi}\inv [ a(g_{n}\expG(\xi),u) -  \AdG_{\expG(\xi)} \inv  a(g_{n},u)  ] \\
        &\quad +  \LambdaG_{\xi}\inv   B(g_{n}\expG(\xi),u)\mu.
    \end{aligned}
\end{align}

We now proceed by extracting the first-order terms in $(\xi, \mu)$ from \eqref{eq:local-error-dynamics-devl}. Since
\begin{align}\label{eq:local-error-dynamics-devl-2}
        \begin{aligned}
            \LambdaG_{\xi}\inv &= I + \tfrac{1}{2} \adG_{\xi}  + \bigO(\lvert \xi \rvert^{2}), \\
            a(g_{n}\expG(\xi),u) &= a(g_{n},u) + \odif*{1} a(g_{n},u)\xi + \bigO(\lvert \xi \rvert^{2}), \\
            \AdG_{\expG(-\xi)} &= \expG(- \adG_{\xi}  ) \\ &= I - \adG_{\xi} + \bigO(\lvert \xi \rvert^{2}),
        \end{aligned}
    \end{align}
it follows that the first line of the right-hand side of \eqref{eq:local-error-dynamics-devl} can be written as $(\odif*{1} a(g_{n},u) - \adG_{a (g_{n},u) })\xi + \bigO(\lvert  \xi   \rvert^{2})$. Furthermore, from smoothness of $B$ and the first order approximation to $\LambdaG_{\xi}$ in \eqref{eq:local-error-dynamics-devl-2}, it follows that 
\begin{align}
    \LambdaG_{\xi}\inv   B(g_{n}\expG(\xi),u)\mu = B(g_{n} ,u)\mu + \bigO(\lvert (\xi, \mu) \rvert^{2} ).
\end{align}
The error dynamics \eqref{eq:local-error-dynamics-devl} can therefore be written as
\begin{align}
    \dot\xi= A(g_{n},u) \xi + B(g_{n},u)\mu + \bigO(|(\xi,\mu)|^{2}),
\end{align}
where $A(g_{n},u) \coloneqq \odif*{1} a(g_{n},u) - \ad_{\LieG}(a(g_{n}),u )$. The leading-order measurement error $y_e$ is obtained similarly,
\begin{equation}
    \begin{aligned}
        y_e 
        &= c(g_n\expG(\xi),u) + D(g_n\expG(\xi),u)\eta - c(g_n,u)\\
        &= C(g_n,u) \xi +  D(g_n, u)\eta + \bigO( \lvert (\xi, \eta) \rvert^{2} ),
    \end{aligned}
\end{equation}
where $C(g_{n},u) \coloneqq \odif* {1} c(g_n,u)$. In conclusion, the error state system can be expressed as
\begin{subequations}\label{eq:local-error-dynamics}
\begin{align}
    \dot{\xi} &= A(g_{n},u) \xi + B(g_{n},u) \mu + \bigO(\lvert (\xi , \mu)\rvert^{2} ),\\
    y_e &= C(g_{n},u)\xi +D(g_{n},u)\eta + \bigO(\lvert (\xi , \eta )\rvert^{2} ).
\end{align}
\end{subequations}
%
When neglecting higher-order terms, the local error state system \eqref{eq:local-error-dynamics} can be characterized as a stochastic linear time-varying system. This allows us to apply the linear Kalman filter \cite{kalman_new_1960}:
\begin{align}
    \begin{aligned}
        \dot{\xi}_{o} &= A(g_{n},u) \xi_o, \\
    \dot{P} &= A(g_{n}, u) P + P A(g_{n},u)\tp + B(g_{n},u) Q B(g_{n}, u)\tp, \\
    \xi_o^{\oplus} &= \xi_o + K(y_e - C(g_{n},u)\xi_{o}), \\  
    P^{\oplus} &= (I - K C(g_{n},u)) P,
    \end{aligned}
    \raisetag{2em}
\end{align}
where $\xi_o$ is the estimated local error state and
\begin{align}
\begin{aligned}\label{eq:filter-gain}
    K &= P C(g_{n},u)\tp (C(g_{n},u) P C(g_{n},u)\tp \\
     & \quad  + D(g_{n},u) N D(g_{n},u)\tp)\inv.
\end{aligned} 
\end{align}

The filter's output would be the optimal estimate of the local error state $\xi$, denoted $\xi_{o}$, which we could add to the nominal trajectory $g_{n}$. This results in a state estimate given by
\begin{equation}
    \label{eq:hat-g-prop}
    g_o\coloneqq g_{n}\expG(\xi_{o}). 
\end{equation}
The connection to the classical linearized Kalman filter \cite[p. 42]{maybeck_stochastic_1982} becomes evident when noting that, for the Lie group $\Rf^k$ with vector addition as the group operation, the estimate simplifies to $g_o= g_{n} + \xi_{o}$.

The EKF is constructed by linearizing the system matrices around the estimate $g_{o}$ once it has been computed.
We proceed as in \cite[Chapter 9.5]{maybeck_stochastic_1982}, assuming the nominal state has been initialized such that our best guess is $\xi_{o} = 0$. Assume now that the filter just has processed measurement $y$ to obtain $\xi_o^\oplus$ and $P^\oplus$, so that our posterior distribution is given by the extended concentrated Gaussian
\begin{equation}
    g^\oplus := g_{n} \expG(\xi^\oplus), 
\end{equation}
where $\xi^\oplus \sim \mathcal{N}(\xi_{o}^\oplus, P^\oplus)$.
As argued for in \cite[Chapter 9.5]{maybeck_stochastic_1982}, the goal is to merge the local error state estimate, $\xi_{o}^\oplus$, into the nominal state. This aims to generate an improved nominal trajectory, thereby strengthening the assumption that deviations from the nominal trajectory remain sufficiently small to justify the use of linear perturbation methods. Consequently, we define the posterior nominal state as
\begin{equation}
\label{eq:g_{n}_update}
    g_{n}^+ := g_o^+ = g_{n} \expG(\xi_o^\oplus).
\end{equation}
The final step is to ensure that the re-parameterized posterior distribution is in some sense close to the original. Denote the re-parameterized state posterior as $g^+ = g_n^+\expG(\xi^+)$ where $\xi^+ \sim \mathcal{N}(\xi_o^+, P^+)$, we try to find the parameters $\xi_o^+, P^+$ minimizing
\begin{equation}
    D_\text{KL}(\mathrm{N}(g^\oplus; g_n,\xi^\oplus_o, P^\oplus)|\mathrm{N}(g^+; g_n^+, \xi^+, P^+)), \label{eq:DKL}
\end{equation}
where $D_\text{KL}(\cdot|\cdot)$ denotes the Kullback-Leibler divergence. This change of reference point in the extended concentrated Gaussian involves relocating the reference point from  $g_{n}$ to $g_{n}^+$. In \cite{ge_geometric_2024}, formulas are established to calculate the mean and covariance of the re-parameterized posterior distribution, $g^+$, that minimizes \eqref{eq:DKL}.
This minimization process yields the optimal parameters
\begin{equation}
    \begin{aligned}
        \xi_{o}^+ &= 0, \quad P^+ = \LambdaG_{\xi_{o}^\oplus}P^\oplus(\LambdaG_{\xi_{o}^\oplus})\tp . \label{eq:P+_formula}
    \end{aligned}
\end{equation}
This step, termed the \emph{covariance reset}, leads to a posterior distribution, $g_{n}^+\exp(\xi^+)$, again a concentrated Gaussian with a local error state estimate equal to zero. Consequently, the local error state estimate is always zero except during the intermediate update step. Therefore, by directly applying the Kalman update to the nominal state $g_{n}$:
\begin{equation}
    \begin{aligned}
    g_{n}^+ &= g_{n} \expG(\xi_o^\oplus) \\
    &= g_{n}\expG(0 + K(y_e - C(g_{n},u)0)) \\
    &= g_{n}\expG(Ky_e) = g_{n}\expG(\zeta),
    \end{aligned}
\end{equation}
where $\zeta \coloneqq Ky_e$, we ensure that the local error state remains zero at all times, rendering it a useless variable for implementation purposes. As a result, the nominal state will always equal the current state estimate, $g_n = g_o$, which is why we propagate and update $g_o$ directly in the LEKF algorithm summarized in \Cref{alg:lie-ekf-left}.

\begin{algorithm}[tb]
	\caption{LEKF from Left-Invariant Error}\label{alg:lie-ekf-left}
	\begin{algorithmic}[0]
        \State \textbf{Propagation:} Integrate the differential equations 
        \begin{align*}
            \dot{g}_{o} &= g_{o}[a_{u}(g_{o})]\twedG\\
            \dot{P} &= A_{u}(g_{o})P + P A_{u}(g_{o})^\top + B_{u}(g_{o}) Q B_{u}(g_{o})^\top
        \end{align*}
        \Statex where
        \begin{equation*}
            A_u(g_{o}) = \odif a_u(g_{o}) - \ad_{\LieG}(a_u (g_{o}) )
        \end{equation*}
        \State \textbf{Measurement update: }
        For a measurement $y$, compute
        \begin{align*}
             K &= PC_u(g_o)\tp (C_u(g_o)PC_u(g_o)\tp + D_u(g_o)ND_u(g_o)\tp)\inv \\
            \zeta &= K (y-c_u(g_o))\\
            g_{o}^{+} &= g_{o} \expG(\zeta)\\ 
            P^{+} &= \LambdaG_\zeta (I - K C_u(g_o)) P \LambdaG_\zeta\tp
        \end{align*}
        \Statex where
        \begin{equation*}
            C_u(g_{o}) = \odif c_u(g_{o})
        \end{equation*}
	\end{algorithmic}
\end{algorithm}

An LEKF can also be derived using the right-invariant error. In the following, we will provide a brief summary of the derivation process for this filter. Quantities associated with the right LEKF will be denoted with a bar. The nominal system model employed is
\begin{subequations}
    \begin{align}
        \dot{\bar{g}}_{n} &\coloneqq [\bar{a}(\bar{g}_{n},u) ]\twedG \bar{g}_{n}, \\
        \bar{y}_{n} &\coloneqq c(\bar{g}_{n},u).
    \end{align}
\end{subequations}
The right-invariant error is defined by $\bar{g}_{e} \coloneqq g \bar{g}_{n}\inv$, and the error dynamics become
\begin{align}
    \dot{\bar{g}}_{e} = [\bar{a}(\bar{g}_{n}\bar{g}_{e},u) - \AdG_{\bar{g}_{e}} \bar{a}(\bar{g}_{n}, u) + \bar{B}(\bar{g}_{n}\bar{g}_{e},u)\mu ]\twedG.
\end{align}
The corresponding local error state is defined by $\bar{\xi} \coloneqq \logG (\bar{g}_{e})$, with dynamics
\begin{align}
   \begin{aligned}
        \dot{\bar{\xi}}  &= \bar{J}_{\bar{\xi}}\inv [\bar{a}(\expG(\bar{\xi})\bar{g}_{n},u) - \AdG_{\expG(\bar{\xi})} \bar{a}(\bar{g}_{n}, u) ] \\
    &\quad +   \bar{J}_{\bar{\xi}}\inv  \bar{B}(\bar{g}_{n}\bar{g}_{e},u)\mu.
   \end{aligned}
\end{align}
Similar considerations as in the derivation of the left-invariant filter then give the leading order terms
\begin{subequations}
\begin{align}\label{eq:right-local-error-dynamics}
    \dot{\bar{\xi}} &= \bar{A}(\bar{g}_{n}, u) \bar{\xi} + \bar{B}(\bar{g}_{n},u)\mu + \bigO( \lvert ( \bar{\xi},\mu) \rvert^{2}), \\ 
    \bar{y}_{e} &= \bar{C}(\bar{g}_{n},u)\bar{\xi} + D(\bar{g}_{o},u)\eta + \bigO( \lvert ( \bar{\xi},\eta) \rvert^{2}),
\end{align}
\end{subequations}
where $\bar{A}(g_{o},u) \coloneqq  \bar{\odif*{}}_{1} \bar{a}(\bar{g}_{n},u) + \adG_{\bar{a}(\bar{g}_{n},u)}$ and $ \bar{C}(\bar{g}_{n},u) \coloneqq \bar{\odif*{}}_{1} c(\bar{g}_{n},u)$. 
The local error state estimate $\bar{\xi}_{o}$ and filter covariance $\bar{P}$ then evolve according to
\begin{align}
    \begin{aligned}
        \dot{\bar{\xi}}_{o} &= \bar{A}(\bar{g}_{n},u) \bar{\xi}_{o}, \\
    \dot{\bar{P}} &= \bar{A}(\bar{g}_{n}, u) \bar{P} + \bar{P} \bar{A}(\bar{g}_{n},u)\tp + \bar{B}(\bar{g}_{n},u) Q \bar{B}(\bar{g}_{n}, u)\tp, \\
    \bar{\xi}_{o}^{\oplus} &= \bar{\xi}_{o} + \bar{K}(\bar{y}_e - \bar{C}(\bar{g}_{n},u)\bar{\xi}_{o}) ,\\  
    \bar{P}^{\oplus} &= (I - \bar{K} \bar{C}(\bar{g}_{n},u)) \bar{P}, 
    \end{aligned} \raisetag{2em}
\end{align}
where the filter gain $\bar{K}$ is computed as in \eqref{eq:filter-gain} from the corresponding barred quantities. The posterior distribution is now given by the extended concentrated Gaussian
\begin{align}
    g^\oplus \coloneqq \expG(\bar{\xi}^\oplus) \bar{g}_{n}, 
\end{align}
where $\bar{\xi}^\oplus \sim \mathcal{N}(\bar{\xi}_{o}^\oplus, \bar{P}^\oplus)$.
The posterior nominal state is now defined by
\begin{align}
    \bar{g}_{n}^+ \coloneqq g_{o}^+ = \exp(\bar{\xi}_{o}^{\oplus}) \bar{g}_n,
\end{align}
and a corresponding new local posterior $\bar{\xi}^+ \sim \mathcal{N}(\bar{\xi}_{o}^+, \bar{P}^+)$ is found by minimizing
\begin{equation}
    D_\text{KL}(\bar{\mathrm{N}}(\bar g^\oplus; \bar g_n,\bar \xi^\oplus_o, \bar P^\oplus)|\bar{\mathrm{N}}(\bar g^+; \bar g_n^+, \bar \xi^+, \bar P^+)).
\end{equation}
 From \Cref{lem:reset-spatial}, it follows that the optimal parameters are given by
\begin{align}
    \bar{\xi}_{o}^+ = 0, \quad \bar{P}^+ = \bar{J}_{\bar{\xi}_{o}^\oplus} \bar{P}^\oplus (\bar{J}_{\bar{\xi}_{o}^\oplus})\tp.
\end{align}




The resulting LEKF derived from the right-invariant error is presented in \Cref{alg:lie-ekf-right}.
\begin{algorithm}[h]
	\caption{LEKF from Right-Invariant Error}\label{alg:lie-ekf-right}
	\begin{algorithmic}[0]
        \State \textbf{Propagation:} Integrate the differential equations 
        \begin{align*}
            \dot{\bar{g}}_{o} &= [\bar{a}_{u}(\bar{g}_{o})]\twedG \bar{g}_{o}\\
            \dot{\bar{P}} &= \bar{A}_{u}(\bar{g}_{o})\bar{P} + \bar{P} \bar{A}_{u}(\bar{g}_{o})\tp + \bar{B}_{u}(\bar{g}_{o}) Q\bar{B}_{u}(\bar{g}_{o})\tp 
        \end{align*}
        \Statex where
        \begin{equation*}
            \bar{A}_{u}(\bar{g}_{o}) =  \bar{\odif*{}} \bar a_u(\bar g_o)+ \ad_{\bar{a}_{u}(\bar{g}_{o})}
        \end{equation*}
        \State \textbf{Measurement update: }
        For a measurement $y$, compute
        \begin{align*}
            \bar{K} &= \bar{P} \bar{C}_{u}(\bar g_{o})\tp(\bar{C}_{u}(\bar{g}_{o}) \bar{P} \bar{C}_{u}(\bar{g}_{o})\tp + D_u(\bar g_o)ND_u(\bar g_o)\tp)\inv \\
            \bar{\zeta} &= \bar{K}(y-c_u(\bar g_o))\\
            \bar{g}_{o}^{+} &= \expG( \bar{\zeta} )  \bar{g}_{o}\\
            \bar{P}^{+} &=  \bar{\LambdaG}_{\bar{\zeta}}  (I - \bar{K}\bar{C}_{u}(\bar{g}_{o}) ) \bar{P} \bar{\LambdaG}_{\bar{\zeta}}\tp
        \end{align*}
        \Statex where
        \begin{equation*}
            \bar{C}_{u}(\bar{g}_{o}) = \bar{\odif*{}} c_u(\bar g_o)
        \end{equation*}
	\end{algorithmic}
\end{algorithm}
\section{Equivalence of the Left- and Right-Invariant Formulations}\label{sec:ekf-equivalence}

This section presents the main theoretical result of this paper. Namely, that the EKFs utilizing a left-invariant and right-invariant error in \Cref{alg:lie-ekf-left} and \Cref{alg:lie-ekf-right}, respectively, both yield the same state estimate. In other words, that the state estimate is invariant with respect to the choice of error. Moreover, we also show that the covariance matrix is related through a change of variables. In order to prove this, we require the following lemma. 
\begin{lem}\label{lem:system-matrix-equiv}
It holds that
    \begin{align*}
        \begin{aligned}
         \bar A(g_{o}, u) &= \AdG_{g_{o}} (A(g_{o}, u) + \adG_{a(g_o, u)} ) \AdG_{g_{o}}\inv \\
         \bar{C}(g_{o}, u) &= C(g_{o}, u)\Ad_{ g_{o} }\inv.
        \end{aligned}
\end{align*}
\end{lem}
\begin{proof}
Utilizing \eqref{eq:relation-right-left-diff}, we have
\begin{align}
    \begin{aligned}
    \bar{\odif*{}} \bar{a}_u(g_{o}) &= \odif*{}  [ \AdG_{  g_o} a_u(g_{o}) ] \AdG_{\bar g_o}\inv \\
    &= \partial_{\xi} [ \AdG_{ g_o} \AdG_{\expG(\xi)} a_u(g_{o} \exp \xi ) ] \rvert_{\xi = 0} \AdG_{  g_o}\inv \\
    &= \AdG_{ g_o}  \partial_{\xi} [ \expG(\adG_{\xi})  a_u(g_{o} \exp \xi ) ] \rvert_{\xi = 0} \AdG_{  g_o}\inv \\
    &= \AdG_{ g_o} \partial_{\xi}[ \expG(\adG_{\xi})  a_u(g_{o}) ] \rvert_{\xi = 0} \AdG_{  g_o}\inv \\
    &\quad + \AdG_{ g_o} \partial_{\xi}[a_u(g_{o}\exp \xi) ] \rvert_{\xi = 0} \AdG_{  g_o}\inv \\
    &=   - \AdG_{  g_o} \adG_{a_u(g_{o})}  \AdG_{  g_o}\inv + \AdG_{{g_{o}}} \odif*{} a_u( g_{o}) \AdG_{  g_o}\inv,  \hspace{-2em}
    \end{aligned}
\end{align}
where the shorthand $\partial_{\xi} = \partial / \partial \xi$ was used. This yields
\begin{align}
\begin{aligned}
    \bar A_u(g_o) &= \bar{\odif*{}}\bar a_u(g_o) + \adG_{\bar a_u(g_o)} \\
    &= \AdG_{{g_{o}}} \odif a_u( g_{o}) \AdG_{g_o}\inv \\
    &= \AdG_{g_{o} } (A_u(g_o) + \adG_{a_u(g_o)} ) \AdG_{ g_{o} }\inv.  
\end{aligned}  
\end{align}   
The relationship between $C$ and $\bar C$ follows from \eqref{eq:relation-right-left-diff}. 
\end{proof}

The following theorem presents the equivalence between the left-invariant and the right-invariant LEKF.
\begin{thm}\label{thm:equivalence}
    The EKFs in \Cref{alg:lie-ekf-left,alg:lie-ekf-right} are equivalent in the sense that they 
    are related 
    through the smooth change of variables 
    \begin{subequations}\label{eq:change-of-variables}
    \begin{align}
        \bar{g}_{o} &= g_{o},  \label{eq:g_o_hat=g_o} \\
        \bar{P} &= \Ad_{g_{o}} P  \Ad_{g_{o}}\tp. \label{eq:P_hat=adPad.T}
    \end{align}
    \end{subequations}
    
\end{thm}
\begin{proof}
We start by showing the equivalence throughout the propagation step. Differentiating \eqref{eq:g_o_hat=g_o} yields
\begin{align}
    \dot{\bar{g}}_{o} &= g_{o}[a ]\twedG = [\Ad_{g_{o}} a ]\twedG g_{o} = [\bar{a} ]\twedG g_{o} = [\bar{a}]\twedG \bar{g}_{o},
\end{align}
which corresponds to the state-propagation of the right-invariant filter. 
By differentiating \eqref{eq:P_hat=adPad.T}, we obtain
\begin{align}\label{eq:P-equivalence-1}
    \begin{aligned}
    \dot{\bar P} 
    &= \Ad_{g_o}\ad_{a  }P\Ad_{g_o}\tp +\Ad_{g_o}P\ad_{a  }\tp\Ad_{g_o}\tp  \\
    &\quad  +\Ad_{g_{o}} (AP + PA\tp + BQ B\tp)   \Ad_{g_{o}}\tp \\
    &= \Ad_{g_o}(A + \ad_{a })P\Ad_{g_o}\tp  + \Ad_{g_o}P(A\tp + \ad_{a  }\tp)\Ad_{g_o}\tp \\
    &\quad + \Ad_{g_o}BQ B\tp\Ad_{g_o}\tp \\
    &= \bar{A}\bar{P} + \bar{P} \bar{A}\tp + \bar{B} Q\bar{B}\tp,
    \end{aligned} \raisetag{2em}
\end{align}
where the last equality follows from \Cref{lem:system-matrix-equiv}. This shows that the transformed covariance on the right-hand-side of \eqref{eq:P_hat=adPad.T} propagates as in the right-invariant filter. We now prove the equivalence over updates. Starting with \eqref{eq:g_o_hat=g_o}, we find
\begin{align}
    \begin{aligned}
        \bar g^+_o  = g_o \expG( \zeta  ) 
         = \expG(\AdG_{g_o} \zeta) g_o . 
    \end{aligned}
\end{align}
Using \Cref{lem:system-matrix-equiv}, the Kalman gains are related by
\begin{align}
    \begin{aligned}
    \bar{K} &= \bar{P} \bar{C}\tp(\bar{C} \bar{P} \bar{C}\tp + DN D\tp)\inv \\
    &= \AdG_{g_o} P \AdG_{g_o}\tp \AdG_{g_o}\invtp C\tp  \\
    & \quad \quad (C \AdG_{g_o}\inv \AdG_{g_o}P \AdG_{g_o}\tp \AdG_{g_o}\invtp C\tp + DND\tp )\tp \\
    &= \AdG_{g_o}PC\tp  (CPC\tp + DND\tp)\inv = \AdG_{g_o}K,
    \end{aligned}
\end{align}
such that $\Ad_{g_{o}}\zeta = \bar{\zeta}$. This shows that the estimates remain the same over an update. An analogous statement holds for the covariance, as is seen from
\begin{align}
    \begin{aligned}
        \bar P^+ &= \Ad_{g_o^+}P ^+\Ad_{g_o^+}\tp \\
        &= \Ad_{g_o}\AdG_{\expG(\zeta)}\LambdaG_\zeta(I - KC ) P \LambdaG_{\zeta}  \tp\AdG_{\expG(\zeta)}\tp\Ad_{g_o}\tp \\
        &=\LambdaG_{(-\Ad_{g_o}\zeta)} \Ad_{g_o}(I - KC)P\Ad_{g_o}\tp\LambdaG_{(-\Ad_{g_o}\zeta )} \tp \\
        &=  \bar{\LambdaG}_{ \bar{\zeta} } (I - \bar{K}\bar{C}) \bar{P} \bar{\LambdaG}_{\bar{\zeta} } \tp,
    \end{aligned}
\end{align}
where \eqref{eq:Adexpxi-Lambdaxi} and \eqref{eq:Adg-Lambda} from \Cref{prop:jac} were used to arrive at the last equality. 
It follows that the filters are related by the change of variables \eqref{eq:change-of-variables}. 
\end{proof}

\Cref{thm:equivalence} implies that both filters result in the same Lie-group-valued state estimate. Consequently, the choice between these filters should be based solely on implementation considerations. Moreover, Monte Carlo simulations in \Cref{sec:case-study} will demonstrate that \Cref{thm:equivalence} does not hold when using a zero-order or a first-order covariance reset approximation. 

\section{Simulation Case Study}\label{sec:case-study}

A simulation study of a spatial navigation problem is conducted to evaluate the equivalence property and the performance of the filters when using zero-order, first-order and full-order covariance resets. The SDEs describing the system are
\begin{equation}
\begin{aligned}
    \dot p^w &= v^w,\\
    \dot v^w &= R^w_b(f_\text{IMU} - b_f - w_f) + \gamma, \\
    \dot R^w_b &= R^w_b[\omega_\text{IMU}-b_\omega-w_\omega]^\wedge_{\SO(3)}, \\
    \dot b_f &= -T^{-1}_{b_f}b_f + w_{b_f}, \\
    \dot b_\omega &= -T^{-1}_{b_\omega}b_\omega + w_{b_\omega},
\end{aligned}
\end{equation}
where $p^w \in \Rf^3, v^w\in \Rf^3, R^w_b\in\SO(3)$ are the position, velocity and orientation of a rigid body with body frame $b$ given in an inertial world frame $w$, $\gamma \in \Rf^3$ is the gravity vector and $(\cdot)^\wedge_{\SO(3)} : \Rf^3\to \so(3)$ denotes the skew-symmetric map. The output of an inertial measurement unit (IMU) placed in the center of mass of the rigid body is simulated by
\begin{equation*}
    f_\text{IMU} = f^b + b_f + w_f \in \Rf^3,\quad
    \omega_\text{IMU} = \omega^b + b_\omega + w_\omega \in \Rf^3,
\end{equation*}
where the specific force $f^b$ is disturbed by a bias $b_f$ and zero-mean Gaussian noise $w_f \sim \mathcal{N}(0,\sigma_f^2 I_3)$ and the angular velocity $\omega^b$ is disturbed by a bias $b_\omega$ and noise $w_\omega\sim\mathcal{N}(0,\sigma_\omega^2I_3)$. The biases are modeled as Gauss-Markov processes with time-constants $T_{b_f}, T_{b_\omega}\in \Rf$, respectively, and driven by Gaussian white noises $w_{b_f} \sim \mathcal{N}(0, \sigma_{b_f}^2I_3)$ and $w_{b_\omega} \sim \mathcal{N}(0, \sigma_{b_\omega}^2I_3)$.

We endow the state with the group structure $g=((R, p, v), b_f, b_\omega) \in \SE_2(3) \times \Rf^3 \times \Rf^3$, as advocated in \cite{barrau_invariant_2014}. 

Loosely coupled global navigation satellite system (GNSS) position measurements are simulated by
\begin{equation}
    y = p^w + \nu \in \Rf^3,
\end{equation}
where $\nu\sim\mathcal{N}(0, \sigma_y^2I_3)$. 

A total number of $M=100$ stochastic 10-second system trajectories are generated, using stochastic specific force and angular velocity inputs resulting in mean and maximum acceleration and angular velocities across the trajectories of $\SI{2.13}{m/s^2}, \SI{8.15}{m/s^2}$ and $\SI{0.16}{rad/s}, \SI{0.49}{rad/s}$, respectively. The initial bias states are drawn from zero-mean Gaussians with $\sigma_{b_f,0} = \SI{0.0073}{\meter\per\second^2}$ and $\sigma_{b_\omega,0} = \SI{0.0012}{\radian\per\second}$, while the rest of the initial system states are set to zero. The IMU is sampled at $1000\text{Hz}$, and a GNSS measurement is sampled every second. The standard deviation for the IMU noises are set to $\sigma_f=\SI{6.9343e-4}{m/s^{3/2}}$, $\sigma_\omega=\SI{3.0853e-5}{\radian/s^{1/2}}$, $\sigma_{b_f}=\SI{4.1881e-5}{m/s^{5/2}}$ and $\sigma_{b_\omega}=\SI{3.9284e-6}{\radian/s^{3/2}}$, where the values align with a STIM300 IMU. The biases time constants are set to $T_{b_f}=T_{b_\omega} = 600\text{s}$. The GNSS measurement is sampled every second with $\sigma_y = \SI{0.07}{m}$. 

 The state is estimated using six filters: the left and right variants of the LEKF using zero-, first-, and full order resets, respectively. A zero-order corresponds to no covariance reset and is equivalent to the IEKF. A first-order reset uses only the first two terms of \eqref{eq:right-jacobian} and \eqref{eq:left-jacobian}, respectively. The full-order filters are named L-FO and R-FO, the first-order L-1O and R-1O, and the zero-order L-0O and R-0O.  

The initial covariance matrix for the left filters is set to $P_0=\diag[(20\si{\degree})^2 I_3, 10^2I_3, 10^2I_3, w_{b_f,0}^2I_3, w_{b_\omega,0}^2I_3]$. The corresponding initial covariance for the right filters is set according to \eqref{eq:P_hat=adPad.T}. The measurement variance is set to $N = 3\sigma_y^2I_3$, where the true simulated GNSS variance is multiplied by 3 to compensate for the errors generated by the linearization step in the Kalman update. The initial estimation state $g_{o_0}$ is sampled from the initial covariance so that the filter initializes in a consistent state. All the different EKFs are initialized with the same initial state. We use the IMU measurements directly and hence run our filter propagation using $\Delta t = \SI{0.001}{s}$.

We define the total error metric between $g_1, g_2 \in \LieG$ as
\begin{equation*}
    \begin{aligned}
    e(g_1, g_2) &\coloneqq |p_1-p_2| + |v_1-v_2| + |\log(R_2\inv R_1)| \\
    &\quad + |b_{f,1}-b_{f,2}| + |b_{\omega,1}-b_{\omega,2}| \geq 0 
    \end{aligned}
\end{equation*}
which is $0$ if and only if $g_1 = g_2$. Suppose that for each experiment indexed $i$ (with $i=1,\dotsc,M$) we have two trajectories $g_{1}^{(i)}, g_{2}^{(i)} : \{1,2,\dotsc, K\} \to \LieG$ where the value at time $k\Delta t$ is written as $g_{1}^{(i)}(k)$ and $g_{2}^{(i)}(k)$. The mean absolute error ($\text{MAE}$) for experiment $i$ is defined by
\begin{equation}
    \text{MAE}(g_{1}^{(i)},g_{2}^{(i)}) \coloneqq  \frac{1}{K} \sum_{k=1}^K e(g_{1}^{(i)}(k), g_{2}^{(i)}(k)).
\end{equation}
Across $M$ experiments, the average $\text{MAE}$ is given by
\begin{equation}
    \overline{\text{MAE}}(\{g_{1}^{(i)}\}_{i=1}^M, \{g_{2}^{(i)}\}_{i=1}^M) \coloneqq \frac{1}{M} \sum_{i=1}^{M} \text{MAE}(g_{1}^{(i)},g_{2}^{(i)}).
\end{equation}
For the special cases of position and orientation, define the position error and the orientation error as
\begin{equation*}
    \begin{aligned}
    e_p(g_1,g_2) &\coloneqq |p_1-p_2|, \quad
    e_R(g_1,g_2) \coloneqq |\log(R_2\inv R_1)|.
    \end{aligned}
\end{equation*}
These definitions naturally extend to corresponding MAE metrics for position and orientation errors by replacing $e(g_1, g_2)$ with $e_p(g_1,g_2)$ or $e_R(g_1, g_2)$ in the averaging formulas above. 

\Cref{tab:sim:filter-diffs} presents the total $\overline{\text{MAE}}$ between each filter. The error between the left and right full-order filters is $0$, indicating that their estimates are identical. However, this property does not hold for the other filters, as evident by the errors between them. Therefore, the choice between the left and right filter will only impact the resulting estimate when using a reduced-order filter. 

\Cref{fig:results:errors} shows the mean and $95$-percentile of the total, position and orientation $\text{MAE}$ for the filters, while \Cref{tab:sim:errors} displays the respective $\overline{\text{MAE}}$ between each filter and the ground truth. The full-order filters exhibit the lowest mean absolute error and have a similarly small orientation error as the left zero-order filter. Generally, the left zero-order filter performs very well in the first few seconds, even outperforming the full-order ones. However, it seems to diverge for a small number of trajectories. This trend appears to be common for the reduced-order filters, suggesting that omitting the covariance shift may lead to inconsistent filters. 
\begin{table}[tb]
    \centering
    \caption{The total $\overline{\text{MAE}}$ between each filter. }
    \begin{tabular}{r@{\hspace{1em}}*{6}{c}}
      & \textbf{L-FO} & \textbf{R-FO} & \textbf{L-1O} & \textbf{R-1O} & \textbf{L-0O} & \textbf{R-0O} \\
    \midrule
    \textbf{L-FO} & 0.00        & 0.00         & 5.39         & 6.08         & 6.20         & 3.00         \\
    \textbf{R-FO} & \greycell{} & 0.00         & 5.39         & 6.08         & 6.20         & 3.00         \\
    \textbf{L-1O} & \greycell{} & \greycell{} & 0.00         & 8.35         & 8.52         & 4.02         \\
    \textbf{R-1O} & \greycell{} & \greycell{} & \greycell{} & 0.00         & 5.87         & 7.45        \\
    \textbf{L-0O} & \greycell{} & \greycell{} & \greycell{} & \greycell{} & 0.00         & 7.50         \\
    \textbf{R-0O} & \greycell{} & \greycell{} & \greycell{} & \greycell{} & \greycell{} & 0.00         \\
    \bottomrule
    \end{tabular}
    \label{tab:sim:filter-diffs}
\end{table}
\begin{figure}[tb]
    \centering
    \begin{subfigure}[b]{0.36\textwidth}
        \centering
        \includegraphics[width=\textwidth, clip, trim=0 8 0 5]{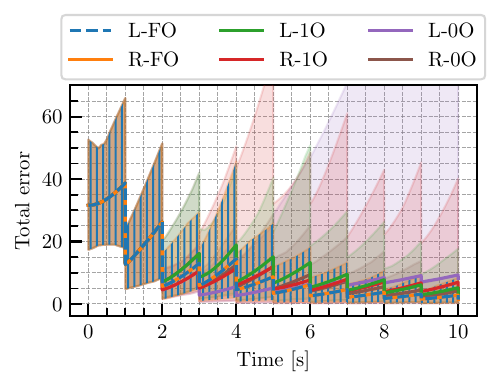}
        \caption{Total}
        \label{fig:results:error:tot}
    \end{subfigure}
    \hfill
    \begin{subfigure}[b]{0.36\textwidth}
        \centering
        \includegraphics[width=\textwidth, clip, trim=0 8 0 5]{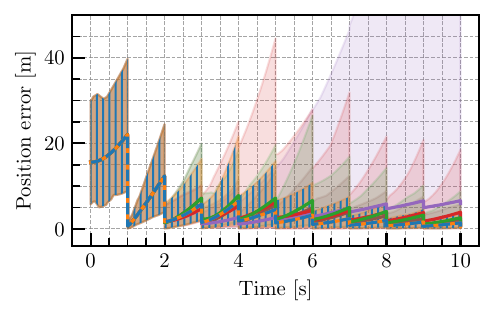}
        \caption{Position}
        \label{fig:results:error:pos}
    \end{subfigure}
    \hfill
    \begin{subfigure}[b]{0.36\textwidth}
        \centering
        \includegraphics[width=\textwidth, clip, trim=0 8 0 5]{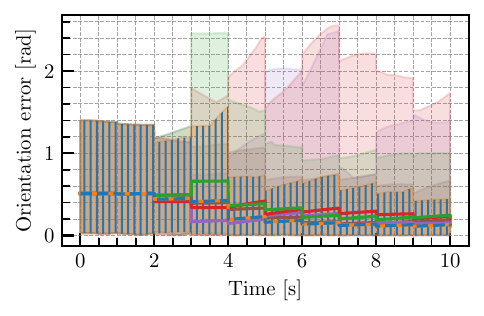}
        \caption{Orientation}
        \label{fig:results:error:ori}
    \end{subfigure}
    \caption{The mean and $95$-percentile of total, position, and orientation $\text{MAE}$ for all filters across $100$ trajectories.}
    \label{fig:results:errors}
\end{figure}%
\begin{table}[tb]
    \centering
    \caption{Total, position and orientation $\overline{\text{MAE}}$ between each filter and ground truth. }
    \begin{tabular}{r@{\hspace{1em}}*{3}{c}}
       & \textbf{Total} & \textbf{Position [m]} & \textbf{Orientation [rad]} \\
    \midrule
    \textbf{L-FO} & 9.37 & 4.12 & 0.28\\
    \textbf{R-FO} & 9.37 & 4.12 & 0.28\\
    \textbf{L-1O} & 11.92 & 5.06 & 0.38\\
    \textbf{R-1O} & 10.52 & 4.90 & 0.35\\
    \textbf{L-0O} & 10.25 & 5.30 & 0.27\\
    \textbf{R-0O} & 10.47 & 4.53 & 0.31\\
    \bottomrule
    \end{tabular}
    \label{tab:sim:errors}
\end{table}%
\section{Conclusions and Future Work}\label{sec:conclusions-and-future-work}
This work demonstrates that employing a full-order covariance reset results in equal state estimates for both the left-invariant and right-invariant extended Kalman filters on matrix Lie groups. Drawing upon established principles of the classical Kalman filter, we derive the filters using a Stratonovich interpretation of the stochastic differential equation. We demonstrate that the re-linearization of the error state model around the posterior estimate naturally includes a covariance reset by recognizing that the re-linearization represents a change of reference point in an extended concentrated Gaussian. We believe that the natural occurrence of the covariance reset during the filter derivation, coupled with our simulation study, suggests that the full-order reset is appropriate for matrix Lie-group-based EKFs.

Future research should investigate the implications of the covariance reset for the convergence results presented for group affine systems using the IEKF. Ultimately, a deeper understanding of Lie group statistics is needed to enhance our comprehension of Lie group estimators.

\bibliographystyle{IEEEtran_nourl}
\bibliography{IEEEabrv,zotero, zotero_finn}

\clearpage

\end{document}